\documentclass[twocolumn,british,pra,superscriptaddress]{revtex4-1}

\pdfoutput=1
\usepackage[utf8]{inputenc}
\usepackage[T1]{fontenc}  
\usepackage{babel}
\usepackage{graphicx}
\graphicspath{ {./Images/} }
\usepackage{amsmath}
\usepackage{amssymb}
\usepackage{url}
\usepackage{mathtools}
\usepackage{graphics,graphicx,color,float}
\usepackage{float}
\usepackage{etoolbox} 
\usepackage{amsfonts}
\usepackage{amsthm, hyperref}
\usepackage[usenames,dvipsnames]{xcolor}
\usepackage{hyperref}
\usepackage{tikz}
\usepackage[T1]{fontenc}
\usetikzlibrary{matrix}
\usepackage{xr-hyper}
\usepackage{tikz}
\usetikzlibrary{calc}

\newcommand{\la}{\langle}
\newcommand{\ra}{\rangle}

\newcommand{\lam}{\lambda}

\newcommand{\R}{\mathbb{R}}
\newcommand{\mcs}{\mathcal{S}}

\newcommand{\C}{\mathbb{C}}

\def\bpm{\begin{pmatrix}}
\def\epm{\end{pmatrix}}

\newcommand{\beq}{\begin{equation}}
\newcommand{\enq}{\end{equation}}
\newcommand{\bel}{\begin{lemma}}
\newcommand{\enl}{\end{lemma}}
\newcommand{\bet}{\begin{theorem}}
\newcommand{\ent}{\end{theorem}}

\newcommand{\tr}{\mathrm{Tr}}
\makeatletter
\newcommand*{\addFileDependency}[1]{
  \typeout{(#1)}
  \@addtofilelist{#1}
  \IfFileExists{#1}{}{\typeout{No file #1.}}
}
\makeatother
 
\newcommand*{\myexternaldocument}[1]{%
    \externaldocument{#1}%
    \addFileDependency{#1.tex}%
    \addFileDependency{#1.aux}%
}

\newcommand{\Tr}{\mathrm{tr}}

\usepackage{booktabs,tabularx}

\newcommand{\suppress}[1]{}

\newcommand{\bra}[1]{\langle #1|}
\newcommand{\ket}[1]{|#1 \rangle}
\newcommand{\braket}[2]{\langle #1|#2\rangle}
\myexternaldocument{SM}

\mathchardef\mhyphen="2D

\def\be{\begin{equation}}
\def\ee{\end{equation}}

\newcommand{\gexcl}{\mathcal{G}_{\mathrm{ex}}}
\newcommand{\h}{\mathcal{H}}

\makeatletter
\newcommand*{\rom}[1]{\expandafter\@slowromancap\romannumeral #1@}
\makeatother

\makeatletter
\appto{\appendix}{%
 \@ifstar{\def\theequation@prefix{A.}}%
 {}%
}
\makeatother

\mathchardef\mhyphen="2D
\newtheorem{theorem}{Theorem}
\newtheorem{lemma}[theorem]{Lemma}

\newtheorem*{main result}{Main Theorem}

\newtheorem*{theorem*}{Theorem}

\begin{document}
	
\title{Local certification of programmable quantum devices of  arbitrary high dimensionality }


\author{Kishor Bharti}
\thanks{equal contribution}
\affiliation{Centre for Quantum Technologies, National University of Singapore}

\author{Maharshi Ray}
\thanks{equal contribution}
\affiliation{Centre for Quantum Technologies, National University of Singapore}

\author{Antonios Varvitsiotis}
\affiliation{Engineering Systems and Design Pillar, Singapore University of Technology and Design}

\author{Ad\'{a}n Cabello}
\affiliation{Departamento de F\'{i}sica Aplicada II, Universidad de Sevilla, E-41012 Sevilla, Spain}
\affiliation{Instituto Carlos I de F\'{\i}sica Te\'orica y Computacional, Universidad de Sevilla, E-41012 Sevilla, Spain}

\author{Leong-Chuan Kwek}
\affiliation{Centre for Quantum Technologies, National University of Singapore} \affiliation{MajuLab, CNRS-UNS-NUS-NTU International Joint Research Unit, Singapore UMI 3654, Singapore}
\affiliation{National Institute of Education, Nanyang Technological University, Singapore 637616, Singapore}
%

\maketitle

{\bf{
The onset of the  era of fully-programmable error-corrected quantum computers will be marked by major breakthroughs in all areas of science and engineering. These devices promise to have significant technological and societal impact, notable examples being the analysis of big data    through better machine learning algorithms  and     the design of  new materials. 
Nevertheless, the capacity  of quantum computers to faithfully implement quantum algorithms  relies crucially on their ability  to prepare
 specific high-dimensional and high-purity quantum states, together with suitable quantum measurements. 
Thus, 
 the  unambiguous certification of these requirements without assumptions on the inner workings  of the quantum computer is critical to the development of trusted quantum processors. One of the most important approaches for  benchmarking quantum devices  is through the mechanism of self-testing \cite{Yao_self,MY98,SB19} that requires a pair of entangled non-communicating quantum devices \cite{reichardt2013classical}. 
 Nevertheless, although computation typically happens in a localized fashion, no local self-testing scheme is known to benchmark high dimensional states and measurements.  
 Here, we show that the quantum self-testing paradigm can be employed to an individual quantum computer that  
is modelled as a programmable black box by introducing  a noise-tolerant certification scheme. We substantiate the applicability of our scheme by providing a family of    outcome statistics whose observation 
certifies 
that the computer is producing  specific high-dimensional quantum states and implementing specific  measurements.}}

\medskip 
Classical computing devices store and manipulate sequences of binary numbers, which are called  bits,   to perform   computational tasks of the utmost social significance   such as healthcare scheduling, weather prediction, and  routing of vehicles. However, as the number of variables  involved in a computational task  
grows beyond a limit, such as in the simulation of $50-60$ spin half particles, even  state-of-the-art supercomputers fail spectacularly  to carry out the required computations. Harnessing the  properties of quantum theory to carry out computations  beyond the reach  of classical  supercomputers is one major objective of quantum information processing. Quantum algorithms  choreograph   the state of a number of qubits (the quantum analogue of classical bits) in an intelligent fashion, which in turn allows to perform highly complex computations. This further requires the error corrected programmable universal quantum computer to be able to generate some specific quantum states and perform specific measurements on them. Such a requirement is a necessary condition and 
must be fulfilled by these powerful futuristic devices. 
\par

In a faithfully minimimalistic scenario, developing  trust in the functionality  of quantum devices necessitates   schemes for certification and benchmarking   that do not rely     on any assumptions  concerning   the inner working of the devices. One of the most  important  approaches for  establishing trust in third-party quantum hardware is via self-testing \cite{popescu1992generic,Yao_self,mckague2012robust,bugliesi2006automata, tsirel1987quantum, summers1987bell}, where the quantum apparatus is modelled as a black box. In this setting,  interactions with the device correspond to  measurements,  
and self-testing leads to   guarantees regarding the uniqueness of the measurement settings and the underlying state preparation, based solely on the measurement statistics.   The first approach to self-testing relied on the use of Bell inequalities, i.e., linear expressions  involving  probabilities of measurement-outcome events that can be carried out in a Bell experiment \cite{Bell64}.  Self-testing  based on Bell experiments is a powerful  technique allowing  to develop insights about the inner working of a pair of non-communicating entangled devices. However, computing typically takes place  in a localized setting and consequently, any  scheme  for benchmarking a progammable quantum computer must be of a local nature  to be of any practical relevance. The first such  local self-testing scheme was introduced  recently in \citep{BRVWCK19}  within the framework of contextuality, a broad generalization of Bell non-locality \cite{KS67, Bell64}. However, the contextuality-based self-testing scheme from \cite{BRVWCK19} was limited in scope due to its applicability to three dimensional quantum systems only. As the dimensional of  the  Hilbert space corresponding to  a quantum computing device grows exponentially as a function of the number of qubits, 
this naturally leads to the problem of identifying 
local certification schemes for arbitrary high dimensionality. 

 \par
In this article, we provide a contextuality-based self-testing scheme for local certification of 
programmable quantum devices of arbitrary high dimensionality. A contextuality scenario is  characterized by a set of measurement events, where two  events are  mutually exclusive if they correspond to same measurement but different outcomes. The exclusivity relations between the measurement  events  can  be conveniently 
 encoded as  edges in an  undirected graph, called the 
 exclusivity graph. 
 The seemingly simple idea to  use  graphs to represent exclusivity relations has  spearheaded  the development of a new line of research at the interface of   graph theory and contextuality~\cite{CSW}. The linear inequalities, the violation of which witnesses contextuality, are referred to as non-contextuality inequalities. Bell inequalities are a special type of non-contextuality inequalities where  the ``contexts'' are  provided via the space-like separation of the parties involved \cite{Bell64, CHSH}. The first ``local'' non-contextuality inequality violated by quantum theory was identified  by Klyachko, Can, Binicio\ifmmode \breve{g}\else \u{g}\fi{}lu and Shumovsky (KCBS)~\cite{KCBS}. The bound on a non-contextuality inequality in a non-contextual hidden variable (NCHV) model is called the NCHV bound. Quantum theory violates the NCHV bound for suitably chosen state and measurement settings and thus manifests as a contextual theory.  
  To  any exclusivity graph we associate a  ``canonical'' non-contextuality inequality.
Furthermore, we say that a graph is self-testable, if the corresponding  non-contextuality inequality admits self-testing.  The main contributions of  this work can be summarized as~follows:

{{\em We give a local and robust self-testing scheme for certifying high-dimensional programmable quantum devices based on contextuality~(see~Figure~\ref{fig:Scheme}).}

 As a key ingredient in the  scheme, we show that the family of odd anti-cyclic graphs with at least five vertices are self-testable.  
  Additionally, we extend the protocol given in  \cite{BRVWCK19} (which only provides sufficient conditions for self-testing), in order to determine if a given graph is provably non-self-testable (see Appendix \ref{counterex}). 
 In particular, we show that not all graphs with a positive gap between NCHV bound and the maximum quantum bound for the corresponding canonical non-contextuality inequality admits self-testing by providing an explicit  counterexample.

 \begin{figure}[h]
\centering
\includegraphics[width=0.48\textwidth]{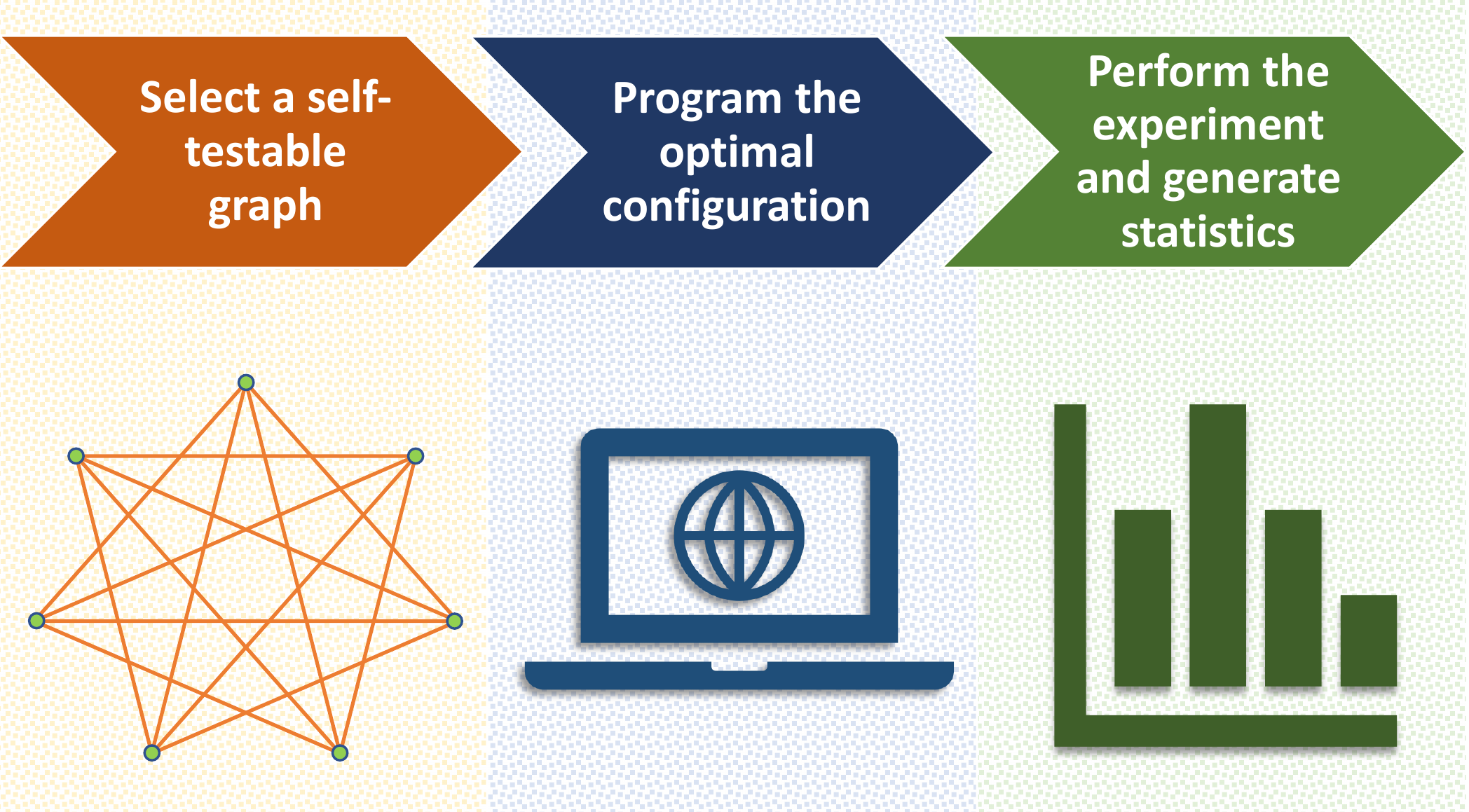}
\caption{The local certification scheme proceeds in three steps. In step $1$, we select a self-testable graph. Then we proceed to program the optimal configuration corresponding to the graph in the quantum computer. In the last step, we carry on the contextuality experiment and generate the measurement statistics. If the measurement statistics correspond to optimal quantum violation, it implies that the device implemented the settings as per the program (upto an isometry). If it reaches close to the optimal violation, the robustness of our scheme still provides confidence about the fidelity of the intended operation with the implemented operation. If it fails (i.e., gives an output outside a pre-established confidence neighbourhood of the extremal value), the protocol indicates that the device is not performing reliably.}
\label{fig:Scheme}
\end{figure}

\medskip 
\par
\noindent {\em Graph approach to contextuality---}
 An arbitrary experimental scenario can be characterized by a set of measurement events $e_1,\ldots,e_n$. Two events are mutually exclusive if they correspond to same measurement but different outcomes. The exclusivity structure of a set of measurement events is captured by the exclusivity graph, denoted  $\gexcl$,  with  nodes $\{1,\ldots,n\}$  (denoted by $[n]$) corresponding  to events  $\{e_i\}_{i=1}^n$. Two nodes  $i$ and $j$ are adjacent (denoted by $i\sim j$) if   $e_i$ and $e_j$ are mutually exclusive.
 
 Given an experimental scenario with exclusivity graph $\gexcl$, a theory assigns probability to the events corresponding to its vertices. The mapping $p: [n]\to [0,1]$, where $p_i+p_j\le 1$, for all $i\sim j$ is called a behaviour. Here, the non-negative numbers $p_i$ refer to the probability of the event $e_i$. We call a behaviour $p$ deterministic non-contextual if all the probabilities $p_i$ are either $0$ or $1$ and the occurrence of a  event does not depend on the possibility of occurrence of other events. The convex hull of all deterministic non-contextual behaviours form a polytope, denoted by ${\cal{P}}_{nc}(\gexcl)$, which contains all possible non-contextual behaviours for the experimental scenario encoded by the graph $\gexcl$. The behaviours which lie outside ${\cal{P}}_{nc}(\gexcl)$ are contextual behaviours. The set of non-contextual behaviours are bounded by finitely many halfspaces, which are called  non-contextuality inequalities. Formally speaking, non-contextuality inequalities correspond to linear inequalities of the form 
 \beq\label{eq:ncineq}
 \sum_{i \in [n]} w_i p_i \leq B_{nc}(\gexcl, w), \ \forall p \in {\cal{P}}_{nc}(\gexcl),
\enq
 where $w_1,\ldots,w_n\ge 0$ and $B_{nc} (\gexcl, w)$ are real scalars. The non-contextuality inequalities with all the weights $\{w_i\text{s}\}$ equal to one will be referred to as canonical non-contextuality inequalities.
By definition, $B_{nc} (\gexcl, w) $ corresponds to the NCHV bound on the linear expression  $ \sum_{i \in [n]} w_i p_i$ and is also equal to the independence number of the exclusivity graph $\gexcl$, defined as  the cardinality of the largest set of pairwise non-adjacent nodes of $\gexcl$ \cite{CSW}.  
 A quantum behaviour has the following  form: 
  \be\label{cwfere}
p_i=\Tr(\rho\Pi_i), \forall i\in [n] \text{ and } \Tr(\Pi_i\Pi_j)=0, \text{ for } i\sim j,
\ee 
 for some quantum state $\rho$ and quantum projectors $\Pi_1,\ldots, \Pi_n$ acting on a Hilbert space $\mathcal{H}$.
 An ensemble $\rho, \{\Pi\}_{i=1}^n$ satisfying \eqref{cwfere} is called a {\em quantum realisation} of the behavior $p$. For a given quantum behaviour $p$, there can be multiple quantum realisations.
The set of quantum behaviours is a convex set, which we  denote by ${\cal{P}}_q(\gexcl)$. The maximum value of the linear expression $\sum_{i \in [n]} w_i p_i$, as   $p$  ranges over the set of quantum behaviors ${\cal{P}}_q(\gexcl)$ can exceed the classical bound. We will denote the maximum attainable quantum value by $B_{cq}(\gexcl, w).$ Interestingly, $B_{cq} (\gexcl,w)$ is equal to the {\em Lov\'asz theta number} of the  graph $\gexcl$ and admits a formulation as a tractable optimisation problem known as a semidefinite program \cite{CSW} (see Methods Section).


\medskip 
\par
\noindent {\em Robust self-testing---}
Informally speaking, a non-contextuality inequality ${\cal{I}}$ is said to self-test a quantum realisation  $\rho, \{\Pi\}_{i=1}^n$ if it achieves the quantum bound for the non-contextuality inequality ${\cal{I}}$ and furthermore, all other quantum realisations which achieve the quantum bound corresponding to ${\cal{I}}$  are equivalent to $\rho, \{\Pi\}_{i=1}^n$ up to global isometry. For a formal definition,  the reader is referred to Section \ref{appA} in the Appendix. As discussed in~\cite{BRVWCK19}, the essential ingredient in proving  self-testing results for a non-contextuality inequality ${\cal{I}}$ with underlying exclusivity graph  $\gexcl$ is that the corresponding  Lov\'asz theta  semidefinite program  (cf.  \eqref{theta:primal}) has an unique optimal solution. In the case of the KCBS inequality,  the exclusivity graph   is a pentagon. The configuration corresponding to optimal quantum violation admits an umbrella structure (see Figure~\ref{fig:umbrella}). The KCBS inequality has been generalized   to odd $n$-cycle exclusivity graphs, which are called $\rm{KCBS}_n$ inequalities. The $\rm{KCBS}_n$ inequalities are the canonical non-concontextuality inequalities for an odd cycle graph and admit robust self-testing \cite{BRVWCK19}.

\medskip 

\par
\noindent {\em  Self-testing anti-cycles---}  Building on the link between graph theory and contextuality \cite{CDLP13},   in combination  with   the  strong perfect graph theorem \cite{strong},  it was shown in \cite{CSW, CDLP13}   that the presence of certain  exclusivity  structures  is a necessary and sufficient condition for a  non-contextuality  scenario to witness quantum contextuality.   These  fundamental exclusivity structures  correspond to an odd number of events that are either cyclically or anti-cyclically  exclusive. 
In the graph theory literature,  odd cycles and odd anti-cycles are  called odd holes and  odd anti-holes respectively. 
\begin{figure}[h]
\centering
\includegraphics[width=0.4\textwidth]{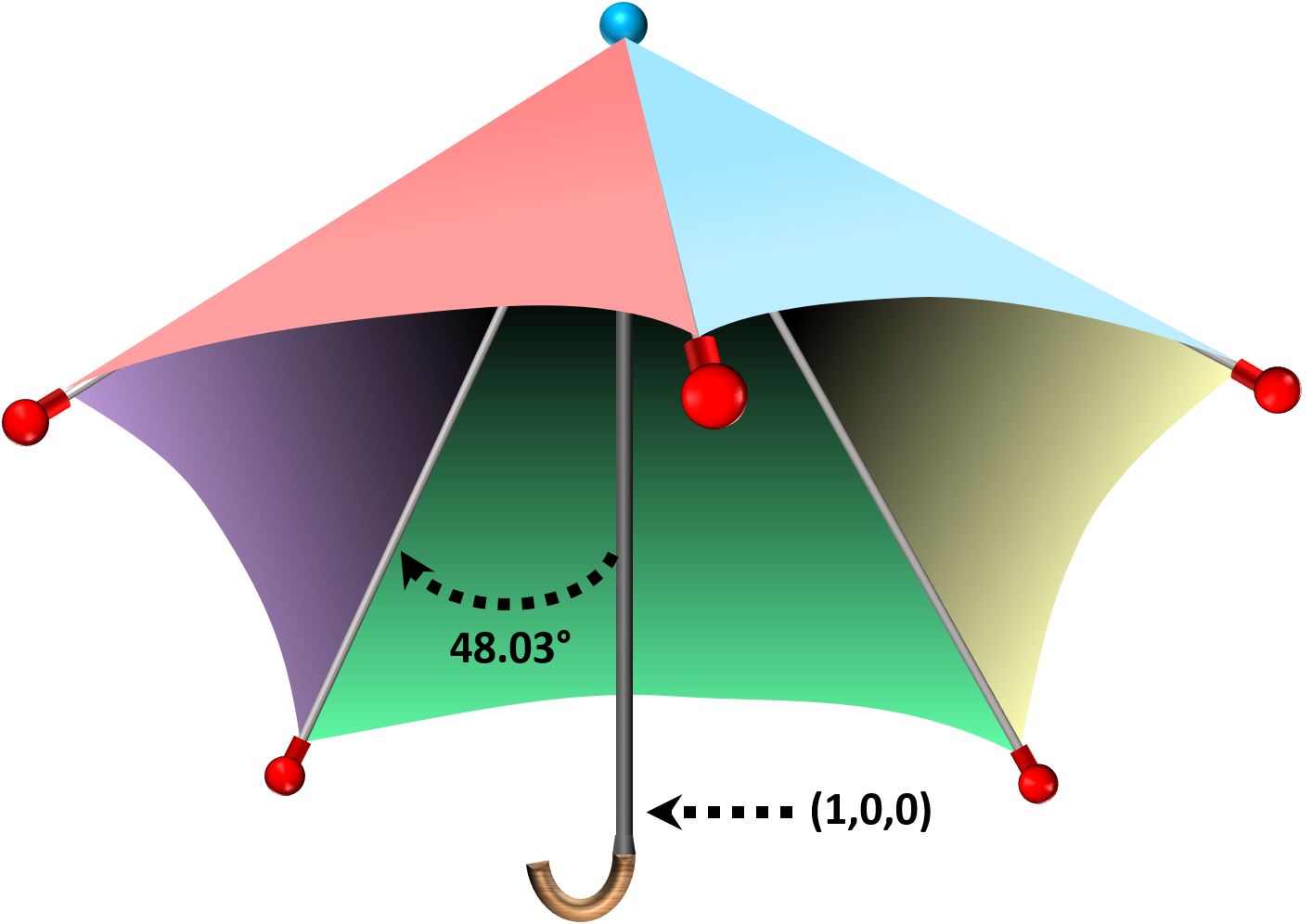}
\caption{The geometry corresponding to the unique optimal configuration of $C_5$ looks like an umbrella. The handle of the umbrella is the quantum state and the vectors corresponding to the projectors are the spokes of the umbrella.}
\label{fig:umbrella}
\end{figure}
The canonical non-contextuality inequality 
for an anti-hole is given by:
\begin{equation}
\label{anti_ineq}
\sum_{i=1}^{n} p_i \leq 2, \ \text{ for all }  p \in {\cal{P}}_{nc} \left(\overline{C_n}\right),
\end{equation} 
which we call 
 anti-hole inequalities 
 and correspond to  facets of the classical polytope for $\overline{C_n}$. The quantum bound for the anti-hole inequalities, i.e.,  the Lov\'asz theta number for $\overline{C_n}$ is ${1+\cos{\pi \over n}}\over \cos{\pi \over n}$. A canonical quantum ensemble which achieves the quantum value for the anti-hole inequalities corresponding to odd $n$ is given by $n-2$ dimensional quantum state and projectors. In fact t{\emph{he dimension of quantum system achieving optimal quantum violation of anti-hole inequalities must be $n-2$ for odd~$n$.}} We provide a proof of the aforementioned claim in  Appendix \ref{appDim}. 
Moreover, we show 
that {\emph{the anti-hole inequalities admit robust self-testing.}} This fact is proven  in  the Appendix (Sections \ref{appA} and \ref{appB}).

 
As  the presence of holes and/or  anti-holes  in a contextuality scenario   dictates 
the  possibility of quantum advantage   \cite{CDLP13}, our result regarding anti-hole self-testing in combination with~\cite{BRVWCK19}  imply that {\em all non-contextuality inequalities which are fundamental to quantum theory admit local robust self-testing}. This is because the generalized KCBS  and anti-hole inequalities are the unique facet-defining non-contextuality inequalities for their respective odd hole and anti-hole exclusivity scenarios~\cite{BAKR18}.  

%

\medskip 

\noindent {\em A new local certification scheme---}Leveraging our self-testing results we propose a local certification scheme that proceeds in three phases  described below:
\begin{enumerate}
\item {\emph{Graph selection:} Choose a  self-testable graph using the algorithm detailed in Figure \ref{fig:flowchart}. In particular, one can use the anti-cycle graphs for benchmarking arbitrary high-dimensional quantum realisations.  }
\item {\emph{Configuration coding:} Solve the Lov\'asz theta SDP corresponding to the graph chosen in Step 1 and compute the unique optimal configuration using the gram vectors of the solution. Next, program the device to create this configuration. }
\item  {\emph{Experimentation:} Perform a contextuality experiment using this configuration and generate the measurement statistics. If the statistics achieve the Lov\'asz theta bound or is within a pre-established confidence neighbourhood, we  trust the device. Otherwise, the protocol indicates that the device is not performing as a reliable and sufficiently non-noisy quantum computer.}

\end{enumerate}

\noindent {\em Conclusion---} One of the central problems in quantum computation is to devise protocols that allow  an efficient classical machine (also known as verifier) to verify the computation done by an efficient quantum machine (also known as prover) \cite{Got04}. 
An important approach has been 
proposed for a weaker variant of the problem where the verifier has access to a small quantum computer of her own, or she has to verify two entangled non-communicating provers \cite{broadbent2009universal, fitzsimons2017unconditionally, aharonov2017interactive, aharonov2008interactive, reichardt2013classical}, which  
 relies on self-testing properties of parallel CHSH games. More importantly, the self-testing nature of parallel CHSH games along with protocols for state and process tomography and computation by teleportation,  \cite{reichardt2013classical} devised a method for realising arbitrary quantum dynamics without assuming anything about the internal structure of the non-communicating quantum devices. In a seminal paper, Mahadev provided an affirmative answer to the problem involving a single verifier and a single prover \cite{mahadev2018classical}. The scheme requires a complexity-theoretic assumption, i.e.,  the verifier has access to post-quantum cryptography, which the efficient quantum prover can not break. Though our scheme depends on the self-testing properties of local non-contextuality inequalities, further work needs to be done for realising arbitrary dynamics in a way similar to \cite{reichardt2013classical}. The classical simulation of the outcome statistics corresponding to our scheme requires systems of higher dimension \cite{KGPLC11,FK17,CGGX18,Budroni19} and has a thermodynamical extra cost \cite{CGGLK16}. The only assumption made in our framework is that measurements are projective, which can be tested experimentally \cite{MZLCAH18}. 

\section*{Methods} \label{sec: method}

The main tool we use in this work  to show that anti-hole inequalities admit robust self-testing is Theorem~\ref{res:uniqueness}, shown in~\cite{BRVWCK19}, which  
provides a sufficient condition for a graph to be self-testable. This result relies crucially 
 on the rich properties  of a powerful class of   mathematical optimisation models, known as 
Semidefinite programs (SDPs) (see Appendix \ref{sec:SDPs}).  SDPs constitute   a vast generalisation of linear optimisation models where scalar variables are  replaced by vectors and the constraints and objective function are affine  in terms of the inner products of the vectors. Equivalently, collecting all pairwise inner products of these vectors in matrix, known as the Gram matrix, an SDP corresponds to optimising a linear function of the Gram matrix  subject to affine constraints. Analogously to linear programs, to any SDP there is an associated  a dual program whose value is equal  to the primal under reasonable assumptions. 
Next, we single out certain properties of primal-dual solutions that are of relevance to his work. A pair of primal dual optimal solutions ($X^*, Z^*$) with no duality gap (i.e. $\Tr(X^*Z^*) =~0$), satisfies  {\em strict complementarity} if the ${\rm Range}(X^*)$ and ${\rm Range}(Z^*)$ give a direct sum decomposition of the underlying space. Furthermore,  an optimal dual solution $Z^*$  with rank $r$ is   {\em dual nondegenerate} if the tangent space at $Z^*$ of the manifold of symmetric matrices with rank equal to $r$ together  with the linear space of matrices defining the SDP span the entire space of symmetric matrices. 
In this work we focus on the Lov\'asz theta SDP, see \eqref{theta:primal} in Appendix~\ref{appA}.
 The proof of our main result  involves two main steps. First, we construct a dual optimal solution of \eqref{theta:dual2} for an odd-cycle graph   by providing an explicit mapping between the Gram vectors of a primal (\ref{theta:primal}) optimal solution of an odd-cycle graph and the Gram vectors of a dual optimal solution of the complement graph; see Theorem~\ref{Const} for details. Next, once we construct a dual optimal solution, we  show in Theorem~\ref{nondegeneracy}  that it satisfies the non-degeneracy conditions given in~\eqref{algcon}.  By Theorem~\ref{res:uniqueness} , this shows that anti-hole inequalities admit robust self-testing. Details of the proofs can be found in~Appendix~\ref{appB}.  Additionally, we show that not all graphs admit self-testing by providing a counter example of such a graph (see Appendix~\ref{counterex}).
   The overall scheme for determining  whether a graph is self-testable  (equivalently whether the primal optimal solution of the Lov\'asz theta SDP corresponding to that graph is unique or not) is provided in the form of a flowchart in Figure~\ref{fig:flowchart}.    
\vspace{0.2cm}
\begin{figure}[h]
\centering
\includegraphics[width=0.5\textwidth]{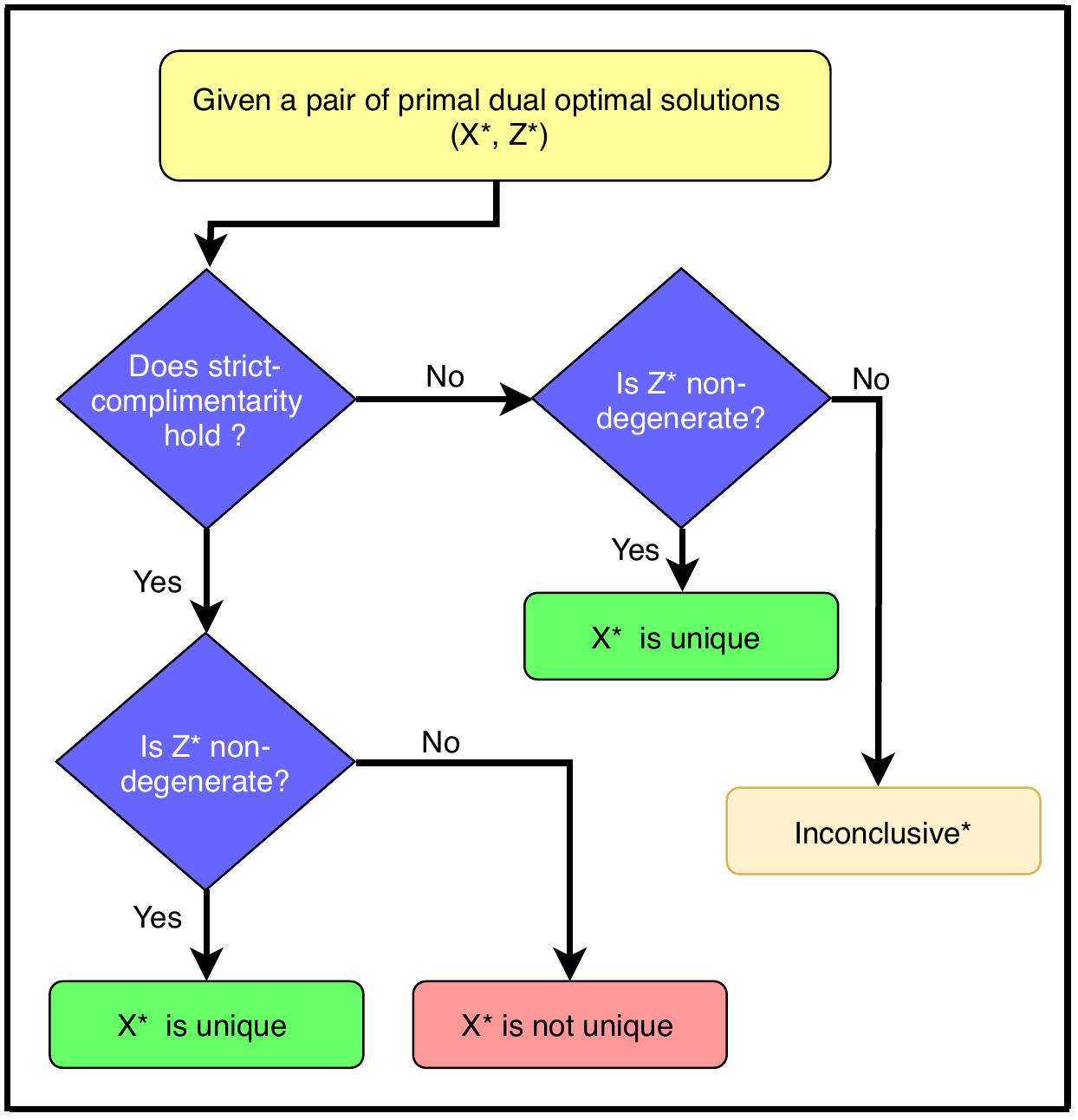}
\caption{Flowchart for determining (non)uniqueness of primal solution(s). The $(*)$ refers to the fact that one can still hope to arrive at a definitive answer by restarting the algorithm using a different dual optimal solution (if it exists).}
\label{fig:flowchart}
\end{figure}

\par 
\noindent {\em Acknowledgements.---}
We thank the National Research Foundation of Singapore, the Ministry of Education of Singapore, MINECO Project No.~FIS2017-89609-P with FEDER funds, and the Knut and Alice Wallenberg Foundation for financial support.
A significant part of the research project  Local certification of programmable quantum devices of arbitrary high dimensionalitywas carried out during ``New directions in quantum information'' conference organized by  Nordita, the Nordic Institute for Theoretical Physics. 
\bibliography{CST.bib}
\bibliographystyle{naturemag}


\appendix

\section{Semidefinite programming  basics}\label{sec:SDPs}
A  semidefinite program (SDP) is given by an optimisation problem of the following form
 \be
 \underset{{ X}}{{\sup}} \left\{ \la C,{ X}\ra : X \in \mcs^n_+,\ \la A_i,{ X} \ra=b_i \ (i\in [m]) \right\},\tag{{P}}\label{primalintro}
 \ee
 where $\mathcal{S}^n_+$ denotes  the cone of $n\times n$ real positive semidefinite matrices and $\la X,Y\ra =\Tr(X^TY)$. The corresponding dual problem is given by
\be
\underset{y,Z}{{\inf}}\left\{ \sum_{i=1}^m b_iy_i\ :\ \sum_{i=1}^my_iA_i-C=Z\in \mcs^n_+\right\}\tag{{D}}\label{dualintro}.
\ee
A pair of primal-dual optimal solutions ($X^*, Z^*$) with no duality gap (i.e. $\Tr(X^*Z^*) =~0$), satisfies  {\em strict complementarity} if 
\begin{equation}\label{SCalg}
{\rm rank}(X^*) + {\rm rank}(Z^*) = n.
\end{equation}
Lastly, an optimal dual solution $Z^*$  is called  {\em dual nondegenerate} 
if 
 the  linear system in the symmetric matrix variable $M$
\[ MZ^* = \tr(MA_1) =\ldots= \tr(MA_m) =0,
\]only admits the trivial solution $M=0$.

Central to this work is the  Lov\'asz theta SDP corresponding to  a graph $G$,  whose primal formulation is:
\be\label{theta:primal}\tag{$P_G$}
\begin{aligned} 
\vartheta(G) = \max & \ \sum_{i=1}^n { X}_{ii} \\
{\rm s.t.} & \ { X}_{ii}={ X}_{0i}, \  i\in [n],\\
 & \ { X}_{ij}=0,\  i \sim j,\\
& \ X_{00}=1,\ X\in \mathcal{S}^{1+n}_+,
\end{aligned}
\ee
and the dual formulation we use is given by:
\be\label{theta:dual2}\tag{$D_G$}
\begin{aligned} 
\vartheta(G) = \min & \ Z_{00} \\
{\rm s.t.} & \ { Z}_{ii}= -(2Z_{0i}+1), \  i\in [n],\\
 & \ { Z}_{ij}=0,\ i \nsim j,\\
& \ Z\in \mathcal{S}^{1+n}_+.
\end{aligned}
\ee

\section{Robust self-testing}\label{appA}
 To prove our main result  we use the following definitions  from~\cite{BRVWCK19}. A non-contextuality inequality  $\sum_{i \in [n]} w_i p_i \leq B_{nc}(\gexcl, w)$ is a {\em self-test}  for the realisation $\{\ket{u_i}\bra{u_i}\}_{i=0}^n$~if:
\begin{enumerate}
\item $\{\ket{u_i}\bra{u_i}\}_{i=0}^n$ achieves the quantum supremum $B_{qc}(\gexcl, w)$;
\item For any other realisation $\{\ket{u_i'}\bra{u_i'}\}_{i=0}^n$ that also achieves $B_{qc}(\gexcl, w)$, there exists an isometry $V$ such that 
\beq\label{eq:st}
V\ket{u_i}\bra{u_i}V^{\dagger}=\ket{u_i'}\bra{u_i'}, \quad 0\le i\le n.
\enq
\end{enumerate}
Furthermore, a non-contextuality inequality $\sum_{i \in [n]} w_i p_i \leq B_{nc}(\gexcl, w)$ is an {\em $(\epsilon, r)$-robust self-test} for $\{\ket{u_i}\bra{u_i}\}_{i=0}^n$ if it is a self-test, and furthermore, for any other realisation $\{\ket{u_i'}\bra{u_i'}\}_{i=0}^n$ satisfying 
$$\sum_{i=1}^n w_i |\braket{u_i'}{u'_0}|^2\ge B_{qc}(\gexcl, w)-\epsilon,$$
there exists an isometry $V$ such~that 
\be\label{eq:rst}
\|V\ket{{u_i}}\bra{{u_i}}V^{\dagger} - \ket{{u'_i}}\bra{{u'_i}}\| \leq \mathcal{O}\left({\epsilon}^r\right), \quad 0\le i\le n. \ee 

The proof of our main result hinges on the following theorem (first introduced in \cite{BRVWCK19}):
\begin{theorem}  \label{res:uniqueness} Consider a  non-contextuality inequality 
$\sum_{i=1}^{n} w_ip_i \le B_{nc}(\gexcl, w)$. Assume that  
\begin{enumerate}
\item There exists an optimal quantum realisation $\{\ket{u_i}\bra{u_i}\}_{i=0}^n$ such that  $$ \sum_i w_i \vert \braket{u_i}{u_0}\vert^2 = B_{qc}(\gexcl, w)$$ and $\braket{u_0}{u_i}\ne 0, $ for all $1\le i\le n$, and
\item There exists   a dual optimal solution $Z^*$ for the SDP \eqref{theta:dual2} such  that the homogeneous linear system 
\be \label{algcon}
\begin{aligned}
M_{0i}  &= M_{ii},  \text{ for all } 1 \leq i \leq n, \\
M_{ij} &= 0,  \text{ for all } i \sim j, \\
MZ^*&= 0,
\end{aligned}
\ee
in the symmetric matrix variable $M$ only admits the trivial solution $M=0$.
\end{enumerate}
Then,  the non-contextuality inequality is an $(\epsilon, {1\over 2})$-robust self-test for $\{\ket{u_i}\bra{u_i}\}_{i=0}^n$.
\end{theorem}
\section{Self-testing anti-hole inequalities}\label{appB}
The anti-hole non-contextuality inequalities are given by 
$\sum_{i=1}^{n} p_i \leq 2$
for all $p \in {\cal{P}}_{nc} \left(\overline{C_n}\right)$. The quantum bound for the anti-hole inequalities,  i.e.,  the Lov\'asz theta number for $\overline{C_n}$ is ${1+\cos{\pi \over n}}\over \cos{\pi \over n}$ \cite{knuth1994sandwich}. A canonical quantum ensemble which achieves the quantum value for the anti-hole inequalities corresponding to odd $n$ is given by $n-2$ dimensional quantum state and projectors.
Explicitly, the quantum state is 
\be \ket{v_0} = (1,0,\cdots,0)^T,
\ee
but the description of the projectors $\{\Pi_j = \ket{v_j}\bra{v_j}\}_{j=1}^n$ is more involved \cite{CDLP13}. Let us denote the $k$-th component of $\ket{v_j}$ corresponding to projector $\Pi_j$ as $v_{j,k}$. For ${0 \leq j \leq n-1}$ and  $0 \leq k \leq n-3$,
\be
v_{j,0} = \sqrt{\vartheta(\overline{C_n})\over n}
\ee
\be
v_{j,2m-1} = T_{j,m} \cos{R_{j,m}} 
\ee
\be
v_{j,2m} = T_{j,m} \sin{R_{j,m}} 
\ee
for $m = 1,2,\cdots {{n-3}\over 2} $ and
\be
T_{j,m} = (-1)^{j(m+1)}\sqrt{2 \cos({\pi \over n}) + (-1)^{m+1} \cos\left({(m+1)\pi \over n}\right)\over {n \cos{\pi \over n}}},
\ee
\be
R_{j,m} = {j(m+1)\pi \over n}.
\ee
For the anti-hole non-contextuality inequalities, the  ensemble described above achieves the quantum value and satisfies the first condition of Theorem \ref{res:uniqueness}. It remains to establish the existence of a dual optimal solution for the SDP corresponding to the Lov\'asz theta number of  anti-hole graphs  such that the conditions in \eqref{algcon} are satisfied. Towards this goal, we first proceed to provide the explicit form of the dual optimal solution.
\begin{theorem}\label{Const}
Let $X^* = {\rm Gram}(v_0,v_1,\cdots,v_n)$ be the  unique optimal solution for $(P_{C_n})$. Then,  
\[ 
Z_n^* = \vartheta(\overline{C_n}){\rm Gram}(-v_0,v_1,\cdots, v_n), 
\]
is a dual optimal solution for $(D_{\overline{C_n}}). $ Another useful expression for $Z^*_n$ is given by:
\be
{ Z}_n^{\star} = \left[ 
\begin{array} {c | c}
\vartheta(\overline{C_n}) & -e^\top\\
\hline 
-e & \mbox{circ}(u)
\end{array}
\right]\in \R^{(1+n)\times (1+n)},
\ee
where $e$ is the  vector of all ones of length  $n$, 
$$u = (1,\vartheta(\overline{C_n})\braket{v_1}{v_2}, \cdots, \vartheta(\overline{C_n})\braket{v_1}{v_n}),$$ and 
 $\mbox{circ}(\cdot)$  maps  an $n$-dimensional vector and outputs the corresponding circulant matrix. 
 \end{theorem}

\begin{proof}It was shown in \cite{BRVWCK19} that $(P_{C_n})$ admits a unique optimal solution $X^*$. For any $k=0,1,\ldots, n-1$,  the map taking $i\to i+1$ (modulo $n$)  is an automorphism of $C_n$ (i.e., a bijective map  that preserves adjacency and non-adjacency). In particular, this implies that $X^*$ is circulant and furthermore, constant along each band. Specifically,  all diagonal entries  of $X^*$ are equal, and as $\vartheta(C_n)=\sum_{i=1}X_{ii}^*$, it follows that 
\be\label{cdsvadb}
\braket{v_i}{v_i}=X^*_{ii}=\vartheta(C_n)/n.
\ee  
Analogously, for a pair of indices $i,j$ with $|i-j|=k$ we have that $X^*_{ij}=\braket{v_1}{v_{k+1}}$.  Moreover, be feasibility of $X^*$ we have  that $X^*_{00} =\braket{v_0}{v_0} = 1$. Thus $Z^*_{00}= \vartheta(\overline{C_n})$ has the correct value and it remains to show that $Z^*$ is feasible. 
Next, by feasibility of $X^*$ we have that   $X^*_{ij}=\braket{v_i}{v_j} = 0$, when $i\sim j$    in $C_n$. Thus, by definition of $Z^*$ we have that $Z^*_{ij} = \vartheta(\overline{C_n})\braket{v_i}{v_j} = 0$ for all edges of $C_n$.   Finally we show that $Z^*_{ii}= -(2Z^*_{0i}+1), \  i\in [n]$. Indeed, 
$$Z^*_{ii}=\braket{v_i}{v_i}\vartheta(\overline{C_n})={\vartheta(C_n)\over n}\vartheta(\overline{C_n})=1,$$
where we used \eqref{cdsvadb} and that 
$\vartheta(C_n)\vartheta(\overline{C_n}) = n$ (see Theorem 8 of \cite{lovasz1979shannon}). 
To finish the proof we note that 

\be
\begin{aligned} 
-(2Z^*_{0i}+1) &= -(2\vartheta(\overline{C_n})\braket{-v_0}{v_i}+1) \\
&= 2\vartheta(\overline{C_n})\braket{v_0}{v_i} - 1\\
&= 2\vartheta(\overline{C_n})\braket{v_i}{v_i} - 1\\
& = 1.
\end{aligned}
\ee
where the second last equality follows from the constraint that $\braket{v_0}{v_i} = \braket{v_i}{v_i}$ for $i \in [n]$ and the last equality follows by substituting $\braket{v_i}{v_i}\vartheta(\overline{C_n}) = 1$.
\end{proof}

%


\noindent


Finally, we show that the dual optimal solution satisfies the conditions in \ref{algcon}.
\begin{theorem}\label{nondegeneracy}
The dual optimal solution $Z^*$, corresponding to the complement of an odd-cycle graph satisfies the conditions in \ref{algcon}.
\end{theorem}

\begin{proof}

We show that for any odd $n,$ the only symmetric matrix $M \in \mathbb{R}^{(1+n) \times (1+n)} $ satisfying 
\be 
M_{00}=0, \ M_{0i}=M_{ii}, \ M_{ij}=0\ ( \forall \, i \sim j), \ MZ^*=0,
\ee
is the matrix $M=0$, where $i \sim j$ here refers to an edge in the $\overline{C_n}$ graph. Barring the $MZ^* =0$ constraint, the rest already guarantee that there are at most $2n$ potentially non-zero entries in the $M$ matrix (not counting the repeated entries) corresponding to $\overline{C_n}$ graph. Let the first row of $M$ be $(0,m_1,m_2,\cdots, m_n)$. We fill the rest of the potential non-zero slots in $M$ with $m_{n+1},m_{n+2},\cdots, m_{2n}$. For example for $n=7$, we have 
\be\label{x_7}
M_7=\left(\begin{array}{c|ccccccc} 
0 & m_1 & m_2 & m_3 & m_4 & m_5 & m_6 & m_7\\
\hline
m_1 & m_1  & m_8&  0& 0 & 0 & 0 & m_{14}\\
m_2 & m_8 & m_2  & m_9 & 0 & 0 & 0 & 0 \\
m_3 & 0 & m_9 & m_3  & m_{10} & 0 & 0 & 0\\
m_4 & 0 & 0 & m_{10} & m_4 & m_{11} &0 & 0\\
m_5 & 0 & 0 & 0 & m_{11} & m_5  & m_{12} & 0  \\
m_6 & 0 & 0 & 0 & 0  & m_{12} & m_6 & m_{13}  \\
m_7 &  m_{14} & 0 & 0 & 0 & 0 & m_{13} & m_7    
\end{array}\right), 
\ee

 For notational convenience, let 
$ MZ^*=\left(\begin{smallmatrix} 
{q}& {r^\top}\\ 
{r} & {T}
\end{smallmatrix}\right)$, where $T \in \mathbb{R}^{n \times n}$,  $q \in \mathbb{R}$ and $r \in \mathbb{R}^{n \times 1}$. 
 In the rest of this section we use the notation $\underline{i}$ to denote $i \mod n,$ where $i$ is an integer.  The linear equation corresponding to $q$ implies that
 \beq
\label{trace}
 \sum_{i=1}^n m_i = 0.
\enq 
For $i \in \{ 1,2,\cdots, n \}$, the $2n$ linear equations corresponding to ${ T}_{i,\underline{i+1}}$ and ${ T}_{i,\underline{i-1}}$ imply 
\beq
\label{equaln+}
m_{n+1} = m_{n+2} = \cdots =m_{2n}.
\enq
Now consider the $n$ linear equations corresponding to $r$. These equations along with \ref{equaln+} imply 
\beq
\label{equaln-}
m_{1} = m_{2} = \cdots =m_{n}.
\enq
Using \ref{equaln-} along with \ref{trace} we have 
\beq
\label{zero-}
m_{1} = m_{2} = \cdots =m_{n} = 0.
\enq
Finally, using the equations corresponding to $r$ again and \ref{zero-} we have 
\beq
\label{zero+}
m_{n+1} = m_{n+2} = \cdots =m_{2n} = 0,
\enq
implying that $M = 0$, as desired. 
\end{proof}

\section{Dimension for optimal violation of anti-hole inequalities} \label{appDim}
\begin{theorem}
Given an anti-hole non-contextuality inequality with an odd number of  $n$ measurement
events, the quantum system achieving the optimal quantum bound must
be at least $(n-2)$ dimensional.
\end{theorem}

\begin{proof}
The value of a  non-contextuality inequality achievable within quantum theory is equal to the Lov\'asz theta number of the underlying graph $G$ and admits the SDP formulation \eqref{theta:primal}. The lower bound on the dimension of a quantum system achieving the
optimal quantum bound  is  the rank of the unique primal optimal 
matrix 
\begin{equation*}
\label{Xmatrix} 
{X}_n^{\star} = 
\left[  \begin{array} 
{c | c} 1 & \frac{\vartheta(\overline{C_n})}{n}e^\top\\ 
\hline  
\frac{\vartheta(\overline{C_n})}{n}e
& \mbox{circ}(u) \end{array} \right],
\end{equation*} 
where $e$ is the all-ones vector of length $n$, $\mbox{circ}(\cdot)$ is the circulant function that takes as input a $n$ dimension vector and outputs a $n \times n$ matrix with  the input vector as its top row and every subsequent row being one place right shifted modulo $n$ and $u = (\frac{\vartheta(\overline{C_n})}{n},\frac{n - \vartheta({C_n})}{2 \vartheta({C_n})^2},0,0,0,\cdots,0,0,\frac{n - \vartheta({C_n})}{2 \vartheta({C_n})^2})$.  
Since   $X_n^{\star}$ is
real, its rank over complex field is the same as over real field,
and equals to the number of nonzero  eigenvalues (with multiplicity).
Furthermore, a lower bound on the rank of $X_n^{\star}$ is given
by the rank of the lower right block matrix (the circulant portion).
The eigenvalues of a circulant matrix can be calculated easily using
the circulant vector. A few lines of algebra yields
the following expression for the eigenvalues of the lower right block
matrix, 
\[
\lambda_{j}=\frac{1}{\vartheta_n}+\frac{n-\vartheta_n}{\vartheta_n}\cos(\frac{2\pi j}{n})
\]
for $j\in[n]$ and $\vartheta_n$ denotes the Lov\'asz theta number
for the holes with odd $n.$ One can see that $\lambda_{j}\neq0$
unless $j=\frac{n-1}{2}$ or $\frac{n+1}{2}.$ Thus, the rank of the
circulant matrix is $n-2$ for all odd values of $n.$ Thus, the lower
bound on the rank of the optimal feasible matrix $X^{\star}$
is $n-2$ which is same as the lower bound on the dimension of the
desired quantum system.
\end{proof}

\section{Complex versus Real SDPs}

\begin{lemma}Consider a  real SDP 
$$ \underset{{ X}} \sup  \left\{ \la C,{ X}\ra : X \in \mathcal{S}^n_+,\ \la A_i,{ X} \ra=b_i \ (i\in [m]) \right\},$$
that admits a unique optimal solution $X^*$ witnessed by a dual nondegenerate  optimal solution  $Z^*$.
Then, the   SDP considered over the complex numbers, i.e., 
\be\tag{$P_\C$}
\underset{{ X}}{\sup} \left\{ \la C,{ X}\ra_\C : X \in \h^n_+,\ \la A_i,{ X} \ra_\C=b_i \ (i\in [m]) \right\},
\ee
still admits a unique  optimal solution, where $\la X,Y\ra_\C=\tr(X^\dagger Y)$ and $\h^n_+$ denotes  the set of $n\times n$ Hermitian positive semidefinite matrices. 

 \end{lemma}
 
 \begin{proof}
First, we  show that the study of a complex SDP can be reduced to an equivalent real SDP. This fact is well known but we provide a brief argument for completeness. Indeed, 
  for any feasible solution $X=X_\R+iX_\C\in \h^n_+$, the constraint $\la A_i, X\ra_\C=b_i$  is equivalent to two constraints on its real and imaginary part, namely: $\la A_i, X_\R\ra=b_i$ and $ \la A_i,X_\C\ra=0$. Furthermore, checking whether $X_\R+iX_\C$  is Hermitian PSD is equivalent to    
$$ \begin{pmatrix}  X_\R  &  -X_\C\\ X_\C & X_\R\end{pmatrix} \in   \mathcal{S}^n_+.$$
Based on these observations we define the realification of $(P_\C)$ as the following SDP over the real numbers:
\be\tag{$P_\R$}
\begin{aligned}
 \sup_{X,Y}  &\  \la C,{ X}\ra  \\
 \text{ s.t.}  & \ \la A_i,{ X} \ra=b_i \ (i\in [m]) \\
 & \ \la A_i,{Y} \ra=0 \ (i\in [m]) \\
 &  \begin{pmatrix} X & -Y\\ Y & X\end{pmatrix} \in \mathcal{S}^{2n}_+.
 \end{aligned}
 \ee
 Clearly, the solutions of $(P_\C)$ are in bijection with the solutions of the  realification, and thus, to show that  $(P_\C)$  has a unique solution it suffices to show that $(P_\R)$ has a unique solution. 
 Bringing $(P_\R)$   into standard SDP form 
 we arrive at the formulation: 
  \be
\begin{aligned}  
 \sup_{W}  &\   \bpm C/2 & 0 \\ 0 & C/2\epm \bullet W \\
 \text{ s.t.}  & \  \bpm  A_i /2& 0 \\0 & A_i/2 \epm \bullet W =b_i \ (i\in [m]) \\
 & \ \bpm 0 & A_i/2 \\A_i/2  & 0\epm \bullet  W=0 \ (i\in [m]) \\
 & X=Z, \quad  Y+Y^T=0,\\
  &  W=\begin{pmatrix} X & Y\\ Y^T & Z\end{pmatrix} \in \mathcal{S}^{2n}_+,
\end{aligned}
 \ee
 whose  dual  is to minimize the function $\sum_{i=1}^m\lam_ib_i$ over all $\lambda_i, \mu_j, t_{ij}, z_{ij}$ satisfying  $$
 \begin{aligned}
 &\sum_{i=1}^m\lam_i \bpm A_i/2 & 0\\0 & A_i/2\epm+\sum_{i=1}^m \mu_i \bpm 0 & A_i/2\\A_i/2 & 0\epm  +\\ 
 & \sum_{ij=1}^nt_{ij} \bpm E_{ij} & 0\\ 0 & -E_{ij}\epm + \sum_{ij=1}^nz_{ij}\bpm 0& E_{ij} \\ E_{ij}^T& 0\epm -\bpm C/2 & 0\\ 0 & C/2\epm  \succeq  0.
 \end{aligned}
 $$
 We conclude the proof by showing that $\bpm Y^* & 0\\ 0 & Y^*\epm$ is a dual nondegenerate optimal solution for the realification. First, by dual feasibility we have   $Y^*=\sum_{i=1}^m y^*_iA_i-C$  for appropriate  scalars $y^*_i$. Setting   $\lam_i=y^*_i$ and  all other dual variables to zero, we have established feasibility. Second, to show optimality note that $\bpm X^*& 0 \\0 & 0\epm$ is optimal for the realification, and furthermore, 
  $\bpm X^*& 0 \\0 & 0\epm \bullet \bpm Y^* & 0\\ 0 & Y^*\epm=~0$.
  Lastly, to check nondegeneracy consider a symmetric matrix $M=\bpm M_1& M_2\\\ M_2^T& M_3\epm$  satisfying
 \be\label{xsdcvf}
 0=\bpm M_1& M_2\\\ M_2^T& M_3\epm\bpm Y^* & 0\\ 0 & Y^*\epm
 \ee
 and 
 \be\label{cweve}
 \begin{aligned}
 0 &=  M \bullet  \bpm A_i/2& 0 \\ 0 & A_i/2\epm =M \bullet  \bpm 0 & A_i/2 \\ A_i/2 & 0\epm= \\ & = M
   \bullet \bpm E_{ij} & 0\\ 0 & -E_{ij}\epm=M\bullet\bpm 0& E_{ij} \\ E_{ij}^T& 0\epm.
   \end{aligned}
 \ee
 Now,  constraint \eqref{xsdcvf} is equivalent to
 $$M_1Y^*=M_2Y^*=M_3Y^*=0.$$
Furthermore, using  \eqref{cweve}, from the third equation we get $M_1=M_3$, from the first one we get $\la M_1, A_i\ra=0$ and from the second one $\la M_2, A_i\ra=0$. Summarizing, for all $k=1,2,3$ we have that 
$$M_kZ^*= 0 \text{ and } \la M_k, A_i\ra=0  \ (i\in [m]).$$
As $Z^*$ is dual nondegenerate  it has the property that   for any  $M \in\mathcal{S}^n$:
$$MZ^*= \la M, A_i\ra=0\ \forall i  \implies M=0.$$
Putting everything together  we get~$M_1=M_2=M_3=0$. 
 
\end{proof}

\section{Not all non-contextuality inequalities admit self-testing}\label{counterex}
We have proved that all fundamental non-contextuality inequalities admit self-testing. A natural question is  whether every non-contextuality inequality with  separation between corresponding non-contextual hidden variable bound and quantum bound admits self-testing.  Below we provide an explicit non-contextuality inequality  which shows  that the answer is negative. 
In graph theoretical terms, we  identify a  non-perfect graph whose  Lov\'asz theta SDP admits multiple primal optimal solutions. We make crucial use of the following result \cite[Theorem 5]{alizadeh} to determine the (non)uniqueness of primal optimal under strict complementarity, see also \cite{thinh2018structure}.

\begin{theorem} \label{SCalizadeh}
Let $(X^*, Z^*)$ be a pair of primal and dual optimal solutions satisfying 
strict complementarity.
Then, uniqueness of $X^*$ implies that $Z^*$ is dual nondegenerate.
\end{theorem}

The exclusivity graph of our counter-example is shown in Figure~\ref{fig:NST_6}. The corresponding canonical non-contextuality inequality is given by
\be
\label{eq:NST_6}
\sum_{i=1}^{6}p_i \leq 2,
\ee
whose  quantum bound  is equal to $\sqrt{5}$. 
\begin{figure}[H]
\centering
\includegraphics[width=0.2\textwidth]{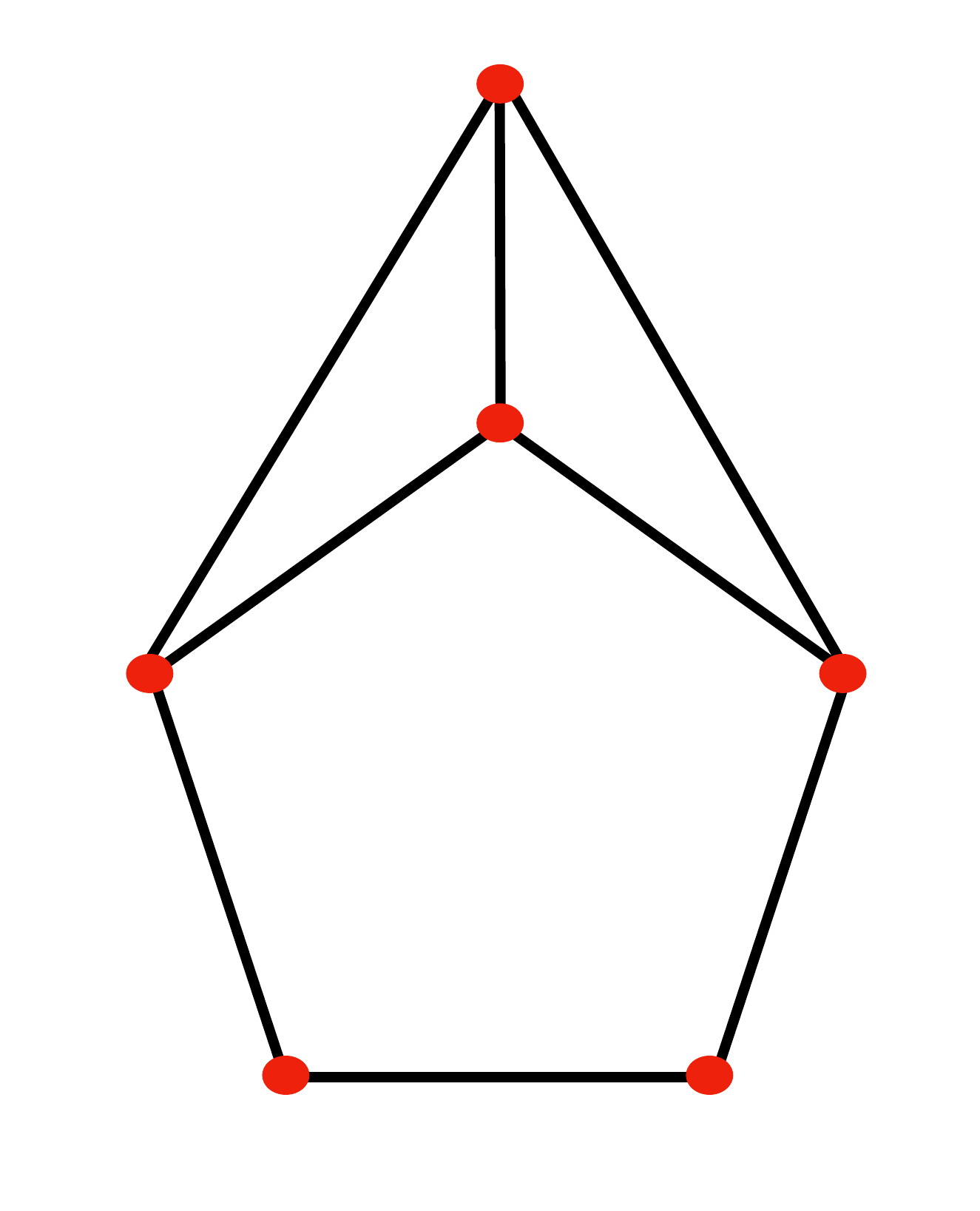}
\caption{The above exclusivity graph corresponds to the canonical non-contextuality inequality with minimal  number of measurement events which doesn't admit self-testing.}
\label{fig:NST_6}
\end{figure}
Consider the   pair of primal-dual optimal solutions 
 \be\label{NST_6_matrix}
{ Z}^{\star} = \left[ 
\begin{array} {c | c  c   c  c  c  c}
\sqrt{5} & -1 & -1 & -1 & -1 & -1 & -1\\
\hline 
-1 &1 & 0 & c & c & 0 & c \\
-1 & 0 & 1 & 0 & c & c & c \\
-1 & c & 0 & 1 & 0 & c & 0 \\
-1 & c & c & 0 & 1 & 0 & 1 \\
-1 & 0 & c & c & 0 & 1 & 0 \\
-1 & c & c & 0 & 1 & 0 & 1 \\
\end{array}
\right],
\ee
where $c = \frac{\sqrt{5}-1}{2}$ and 
 \be\label{NST_6_primal}
{ X}^{\star} = \left[ 
\begin{array} {c | c  c   c  c  c  c}
1 & f & f & f & h & f & h\\
\hline 
f &f & k & 0 & 0 & k & 0 \\
f & k & f & k & 0 & 0 & 0 \\
f & 0 & k & f & r & 0 & r \\
h & 0 & 0 & r & h & r& 0 \\
f & k & 0 & 0 & r & f & r \\
h & 0 & 0 & r & 0 & r & h \\
\end{array}
\right],
\ee
where $f = \frac{1}{\sqrt{5}}$, $h = \frac{f}{2}$, $k = \frac{5-\sqrt{5}}{10}$ and $r = \frac{k}{2}$. Since ${\rm rank}(Z^*)=3$ and ${\rm rank}({X})^{\star}=4$,
 strict complementarity holds. Using Theorem \ref{SCalizadeh}, the uniqueness of $X^*$   implies dual nondegeneracy. To determine dual nondegeneracy for $Z^*$ 
 we (once again) resort to solving system of linear equations.
The symmetric variable matrix $M$ is given by
 \be\label{NST_6_M}
{ M} = \left[ 
\begin{array} {c | c  c   c  c  c  c}
0 & m_0 & m_1 & m_2 & m_3 & m_4 & m_5\\
\hline 
m_0 &m_0 & m_6 & 0& 0 & m_7 & 0 \\
m_1 & m_6 & m_1 & m_8 & 0 & 0 & 0 \\
m_2 & 0 & m_8 & m_2 & m_9 & 0 & m_{10} \\
m_3 & 0 & 0 & m_9 & m_3 & m_{11} & 0 \\
m_4 & m_7 & 0 & 0 & m_{11} & m_4 & m_{12} \\
m_5 & 0 & 0 & m_{10} & 0& m_{12} & m_5 \\
\end{array}
\right],
\ee
Solving for the linear systems of equations ${ Z}^{\star} M = 0$, we get $m_{11} = -m_{12}$, 
$m_{10} = m_{12}$, 
$m_5 =  \frac{1+\sqrt{5}}{2} m_{12}$, 
$m_9 = -m_{12}$, 
$m_3 = -\left(\frac{1+\sqrt{5}}{2}\right)m_{12}$, 
$m_0 = m_1 = m_2 = m_4 = m_6 = m_7 = m_8  = 0$. For example, if we set $m_{12} = 1$, we can get a consistent assignment of $m_i$, from $i=0$ to $12$, which isn't all zero. Hence, the dual solution $Z^*$ is degenerate, which together with strict complementarity implies that the primal is not unique. Thus the non-contextuality inequality in \eqref{eq:NST_6} does not  admit self-testing. 

We also report that we found several other non-perfect graphs (and equivalently non-contextuality inequalities) which do not admit self-testing. Identifying the exact classes of graphs which admit self-testing will be interesting but we leave that as an open question. 

\end{document}